\documentclass[a4paper, 11pt]{article}

\usepackage[utf8]{inputenc}
\usepackage{hyperref}
\usepackage{amsmath}
\usepackage{amssymb}
\usepackage{authblk}
\usepackage{mathtools}
\usepackage{lipsum}
\usepackage{geometry}
\usepackage{todonotes}
\newenvironment{proof}{\paragraph{Proof:}}{\hfill$\square$}
\newtheorem{theorem}{Theorem}[section]
\numberwithin{theorem}{section}

\newtheorem{problem}[theorem]{Problem}

\newcommand{\F}{\mathbb{F}}
\newcommand{\I}{\mathcal{I}}
\newcommand{\M}{\mathcal{M}}
\newcommand{\N}{\mathbb{N}}

\DeclarePairedDelimiter\floor{\lfloor}{\rfloor}

\begin{document}
	\pagestyle{empty}
	
	\title{On Maximal Families of Binary Polynomials \\ with Pairwise Linear Common Factors}
	\date{}
	\author[*]{Maximilien Gadouleau}
	\author[**]{Luca Mariot}
	\author[**]{Federico Mazzone}
	\affil[*]{\small Department of Computer Science, Durham University, Durham, United Kingdom}
	\affil[**]{\small Semantics, Cybersecurity and Services group, University of Twente, Enschede, The Netherlands}
	
	\maketitle
	\begin{abstract}
		We consider the construction of maximal families of polynomials over the finite field $\F_q$, all having the same degree $n$ and a nonzero constant term, where the degree of the GCD of any two polynomials is $d$ with $1 \le d\le n$. The motivation for  this problem lies in a recent construction for subspace codes based on cellular automata. More precisely, the minimum distance of such subspace codes relates to the maximum degree $d$ of the pairwise GCD in this family of polynomials. Hence, characterizing the maximal families of such polynomials is equivalent to determining the maximum cardinality of the corresponding subspace codes for a given minimum distance. We first show a lower bound on the cardinality of such families, and then focus on the specific case where $d=1$. There, we characterize the maximal families of polynomials over the binary field $\F_2$. Our findings prompt several more open questions, which we plan to address in an extended version of this work.
	\end{abstract}
	
	\thispagestyle{empty}
	
	\section{Background and Problem Statement}
	\label{sec:background}
	In what follows, we denote by $\F_q$ the finite field of order $q$, with $q$ being a power of a prime number, while $\F_q^n$ represents the $n$-dimensional vector space over $\F_q$. Further, $\F_q[x]$ denotes the the ring of polynomials with coefficients in $\F_q$ in the unknown $x$. 
	
	For all $n \in \N$, we define $S_n$ as follows:
	$$S_n := \{ f \in \F_q[x] : \deg(f) = n, f \,\text{monic}, f(0) \ne 0 \} \enspace .$$
	In other words, $S_n$ is the family of all monic polynomials in $\F_q[x]$ of degree $n$ and with a nonzero constant term. For any $k$, Let $\I_k \subseteq S_k$ be the set of irreducible polynomials of degree $k$, and let $I_k := |\I_k|$, which can be computed through \emph{Gauss's formula}~\cite{gauss-irr}.
	
	Furthermore, given $d \in \{0, \dots, n\}$, let us define $\M_n^d \subseteq \mathcal{P}(S_n)$ as:
	$$\M_n^d := \{ R \subseteq S_n : \forall f \ne g \in R, \deg(\gcd(f,g)) \le d \} \enspace .$$
	Thus, $\M_n^d$ is a family of subsets of polynomials in $S_n$, such that the degree of the GCD of any two distinct polynomials in it is at most $d$.
	
	We are interested in the following problem:
	\begin{problem}
		\label{pb:state}
		Let $n \in \N$ and $d \in \{0,\cdots, n\}$. What is the size of the largest subset in $\M_n^d$?
	\end{problem}
	
	The motivation for studying Problem~\ref{pb:state} stems from the construction of \emph{subspace codes} in the domain of \emph{network coding}~\cite{medard2011}. Subspace codes are essentially families of vector subspaces of $\F_q^n$. The distance between any two subspaces is defined as the sum of their dimensions, minus twice the dimension of their intersection~\cite{koetter08}. A general research question is then to investigate the trade-off between the cardinality of a subspace code $\mathcal{C}$ and its minimum distance $d_{min}$, i.e. the minimum distance between any two subspaces belonging to $\mathcal{C}$.
	
	Recently, the second and the third author of this abstract proposed in~\cite{mm-automata2023} a new construction of subspace codes based on linear cellular automata (CA), which can be seen as a specific kind of linear mappings $F: \F_q^n \to \F_q^m$ with \emph{shift-invariant} output coordinates, that can also be defined by polynomials in $\F_q[x]$. An interesting finding of~\cite{mm-automata2023} is that the minimum distance of these subspace codes is related to the maximum degree of the GCD of any two polynomials defining the corresponding families of CA; in particular, the higher the degree, the smaller is the minimum distance of the resulting code.
	
	Problem~\ref{pb:state} has already been solved in~\cite{mm-automata2023} for the case $d=0$, by leveraging the construction of pairwise coprime polynomials exhibited in~\cite{mariot20}. In fact, this construction results in the \emph{partial spread code} that the authors of~\cite{gadouleau23} used to define a new family of bent functions.
	
	The contributions of this extended abstract	are as follows:
	\begin{itemize}
		\item We first show a lower bound on the cardinality of the maximal families in $\M_n^d$.
		\item Then, we build a maximal family for the specific case of $q=2$ and $d=1$.
	\end{itemize}
	Remark that the case $q=2$ corresponds to the construction of subspace codes with binary linear CA, which can be seen as a specific kind of (linear) vectorial Boolean functions. This represents also an interesting case from the point of view of the practical applications to network coding.
	
	
	\section{Lower Bound in the General Case}
	\label{sec:lower-bound}
	
	We first start by proving a lower bound on the cardinality of the maximal families in $\M_n^d$.
	Consider the family of polynomials in $S_n$ built as follows.
	Note that we work under the implicit assumption that $d < n / 2$.
	The following construction can be easily adapted if this assumption does not hold.
	
	\noindent \, \textsc{Construction-Lower-Bound}$(n, d)$
	\begin{enumerate}
		\item Take all irreducible polynomials of degree $n$, namely all the elements in $\I_n$.
		\item For all $i \in \{ 1, \dots, d\}$, for all $h \in \I_{n-i}$, pick a $g \in \I_i$ and take $gh$.
		\item For all $i \in \{ d+1, \dots, \floor{(n-1)/2}\}$, for all $g \in \I_i$, pick a $h \in \I_{n - i}$ not previously used and take $gh$.
		\item If $n$ is even, for all $g \in \I_{n/2}$, take $g^2$.
		\item For all $i \in \{ 1, \dots, d\}$, for all $g \in \I_i$, pick a $h \in \I_{n - \floor{n/i} i}$ and take $g^{\floor{n/i}}h$.
	\end{enumerate}
	
	In steps 1, 3, and 4, we proceed as in the coprime case studied in~\cite{mariot20}.
	In particular, we take all irreducible polynomials as they always fit the condition to be in a family of $\M_n^d$.
	Then, we combine each irreducible of degree $d + 1$ with a distinct irreducible of degree $n - d - 1$, and so on also for degrees $d + 2, \dots, \floor{(n-1)/2}$.
	It is always possible to find these distinct polynomials as $I_i$ is monotonically non-decreasing in $i$ (see Lemma 3 of~\cite{mariot20}), and thus $I_{n - i} \ge I_i$ for $i \le n / 2$.
	However, in this case we are allowed to have a common factor of degree at most $d$.
	Hence, when considering combinations of the form $gh$ with $g \in \I_i$ and $h \in \I_{n-i}$ for $i \le d$ (in step 2), we can pick the elements in $\I_i$ multiple times, as their degree is at most $d$.
	This allows us to index these combinations over the irreducibles in $\I_{n-i}$, which are more in number than the ones in $\I_i$.
	Moreover, we can get all the smooth combinations of small factors of degrees $\le d$.
	To make sure to avoid big common factors, in step 5 we use powers of individual irreducibles of degree $\le d$, combined with a suitable $h$ of degree $< d$.
	
	Hence, we have shown that the set generated by the \textsc{Construction-Lower-Bound}$(n, d)$ is indeed a member of the family $\M_n^d$.
	The cardinality of such set is given by
	\begin{equation*}
		\sum_{i = 1}^{\floor{n/2}}{I_i} + \sum_{i = n-d}^{n-1}{I_i} + I_n \enspace .
	\end{equation*}
	The first sum is due to steps 3, 4, and 5, the second sum is due to step 2, while $I_n$ is due to step 1.
	We have thus obtained a lower bound for the cardinality of the maximal family in $\M_n^d$.
	In the next section, we show that this family is maximal at least in a specific case.
	
	
	\section{Maximal Family for Linear Common Factor}
	
	We now show that the construction proposed in the previous section is maximal when considering it on the binary field $\F_2$, for the specific case of $d = 1$.
	This means allowing for a common factor of degree at most 1, namely $1$ or $x + 1$.
	Let us first adapt the construction for this case.
	
	\vspace{12pt}
	
	\noindent \, \textsc{Construction-Maximal}$(n)$
	\begin{enumerate}
		\item Take all irreducible polynomials of degree $n$, namely all the elements in $\I_n$.
		\item For all $g \in \I_{n-1}$, take $(x + 1) g$.
		\item For all $i \in \{ 2, \dots, \floor{(n-1)/2}\}$, for all $g \in \I_i$, pick a $h \in \I_{n - i}$ not previously used and take $gh$.
		\item If $n$ is even, for all $g \in \I_{n/2}$, take $g^2$.
		\item Take $(x + 1)^n$.
	\end{enumerate}
	
	\vspace{8pt}
	
	As shown in Section \ref{sec:lower-bound}, the set produced by \textsc{Construction-Maximal}$(n)$ belongs to the family $\M_n^1$.
	Moreover, it has cardinality
	$$\sum_{i = 1}^{\floor{n/2}}{I_i} + I_{n - 1} + I_n \enspace .$$
	We show that this set is also maximal as follows.
	
	\begin{theorem}
		In $\F_2$, a maximal element of $\M_n^1$ has cardinality
		$$N_n := \sum_{i = 1}^{\floor{n/2}}{I_i} + I_{n - 1} + I_n \enspace .$$
	\end{theorem}
	\begin{proof}
		Let $A \in \M_n^1$ be a maximal element. We know by construction that there exists an element in $\M_n^1$ with cardinality $N_n$, meaning that $|A| \ge N_n$. Clearly, $A$ must contain all irreducible polynomials of degree $n$, $\I_n$. Let $B := A \setminus \I_n$. Given a polynomial $f \in \F_q[x]$, let $l(f)$ be its irreducible factor of lowest degree (if more than one are present, pick the first one in lexicographical order). For any $f \in B$, we define the function
		$$L(f) :=
		\begin{cases*}
			l(f) & if $x + 1 \nmid f$ \\
			l(f / (x + 1)) & if $x + 1 \mid f$
		\end{cases*}
		$$
		First, we prove that the image of $L$ is a subset of $\bigcup_{i=1}^{\floor{n/2}}{\I_i} \cup \I_{n-1}$.
		In the first case of the definition of $L$, since $f$ is reducible we have that $\deg(l(f)) \le \floor{n/2}$ (also $\deg(l(f)) \ge 2$). In the second case, let $g$ such that $f = (x + 1) g$, then we have the following two cases:
		\begin{itemize}
			\item if $x + 1 \mid g$, then $l(f / (x + 1)) = l(g) = x + 1$;
			\item if $x + 1 \nmid g$, then we consider the following cases:
			\begin{itemize}
				\item if $g$ is irreducible, then $l(g) = g$, thus $\deg(l(g)) = n - 1$;
				\item if $g$ is reducible, then $\deg(l(g)) \le \floor{(n-1)/2}$ (also $\deg(l(g)) \ge 2$).
			\end{itemize}
		\end{itemize}
		Thus the image of $L$ is a subset of $\bigcup_{i=1}^{\floor{n/2}}{\I_i} \cup \I_{n-1}$.
		
		Now, we prove that $L : B \to \bigcup_{i=1}^{\floor{n/2}}{\I_i} \cup \I_{n-1}$ is injective. Let $f_1 \ne f_2 \in B$ such that $L(f_1) = L(f_2)$.
		\begin{itemize}
			\item If $\deg(L(f_1)) > 1$, then $f_1$ and $f_2$ have a common factor of degree higher than $1$, leading to contradiction.
			\item If $\deg(L(f_1)) = 1$, then $L(f_1) = x + 1$, which can only be possible if we are in the second case of the definition of $L$, namely $x + 1 \mid f_1$. But then $f_1 = (x + 1)^2 g_1$ for some $g_1$ and the same would hold for $f_2$. Thus, $f_1$ and $f_2$ would share at least a quadratic factor $(x + 1)^2$, leading to contradiction.
		\end{itemize}
		Since $L$ is injective, we have that $|B| \le \sum_{i = 1}^{\floor{n/2}}{I_i} + I_{n - 1}$ and $|A| \le \sum_{i = 1}^{\floor{n/2}}{I_i} + I_{n - 1} + I_n = N_n$, hence $|A| = N_n$.
	\end{proof}
	
	\paragraph{Characterization of maximal families}
	Note that in the proof above we also got that $L : B \to \bigcup_{i=1}^{\floor{n/2}}{\I_i} \cup \I_{n-1}$ is a bijection. We can use this function to characterize all maximal families in $\M_n^1$ as follows:
	\begin{itemize}
		\item Obviously $\I_n$ is part of any maximal family.
		\item Assuming $n > 2$. If $n$ is even and $g \in \I_{n/2}$, then $L^{-1}(g) = gh$ for some $h \in \I_{n/2}$. If $h \ne g$ then $L^{-1}(h) \ne gh = L^{-1}(g)$, and $\gcd(L^{-1}(h), L^{-1}(g)) = h$ and its degree is greater than $1$, leading to contradiction. Thus $g = h$.
		\item For $i \in \{1, \dots, \floor{(n-1)/2}\}$, if $g \in \I_i$, then we have four cases:
		\begin{itemize}
			\item (only applicable if $i \mid n$) $L^{-1}(g)=g^a$ with $a = n / i$
			\item (only applicable if $i \mid n - 1$) $L^{-1}(g)=(x+1) g^a$ with $a = (n - 1) / i$
			\item $L^{-1}(g)=g^a h$ for some $a \ge 1$, $ai < n$, $g \nmid h$, $\deg(h) > 1$, and $x + 1 \nmid h$. Note that $h$ must be irreducible and $ai < n/2$, in particular $h \in \I_{n - ai}$.
			\begin{proof}
				If $h$ is reducible or if $ai < n / 2$, then $\deg(l(h)) \le \floor{n/2}$, thus we can apply $L^{-1}$ to $l(h)$. Then $l(h)$ would be a common divisor of $L^{-1}(g)$ and $L^{-1}(l(h))$. Note that $l(h) \ne g$ since $g \nmid h$ and thus $\gcd(g, h) = 1$, hence $L^{-1}(g) \ne L^{-1}(l(h))$. Note that $\deg(l(h)) > 1$ since $x + 1 \nmid h$. Thus we have a contradiction.
			\end{proof}
			
			Also note that $h$ does not divide any other $f' \in A$. Otherwise we would have a common factor of degree strictly greater than $1$.
			\item (with $\deg(g) = i \ge 2$) $L^{-1}(g)=(x + 1) g^a h$ for some $a \ge 1$, $ai < n$, $x + 1 \nmid g$, $g \nmid h$, $\deg(h) > 1$, and $x + 1 \nmid h$. Note that $h$ must be irreducible and $ai - 1 < n / 2$, in particular $h \in \I_{n - ai - 1}$.
			\begin{proof}
				Similar to the previous one.
			\end{proof}
		\end{itemize}
	\end{itemize}
	Hence, we have proven that a set $A \in \M_n^1$ is maximal if and only if:
	\begin{itemize}
		\item $A$ contains $\I_n$;
		\item if $n$ is even then $A$ contains $\{g^2: g \in \I_{n/2}\}$;
		\item for every $g \in \I_i$ with $1 < i < n / 2$ there exists a unique $f \in A$ such that $g \mid f$, and this $f$ is either of the form
		\begin{itemize}
			\item $f = g^a$ with $a = n / i$, or
			\item $f = (x + 1) g^a$ with $a = (n - 1) / i$, or
			\item $f = g^a h$ with $ai < n / 2$ and $h \in \I_{n - ai}$, or
			\item $f = (x + 1) g^a h$ with $ai < n / 2 + 1$ and $h \in \I_{n - ai - 1}$;
		\end{itemize}
		and in the last two cases $h$ does not divide any other $f' \in A$;
		\item either $(x + 1)^n \in A$ or $(x + 1)^a h \in A$ for $1 < a < n / 2$ and some $h \in \I_{n - a}$.
	\end{itemize}

	\section{Conclusions and Future Works}
	\label{sec:open}
	In this abstract, we considered the problem of characterizing the maximal families of polynomials of degree $n$ and nonzero constant term over $\F_q$ where the degree of the GCD of any two polynomials is at most $d$. The motivation for studying this problem originates from a recent construction of subspace codes based on linear cellular automata. After providing a general lower bound, we focused our attention on the specific case where $q=2$ and $d=1$, i.e. each pair of polynomials in the family has a linear common factor. We proved the formula for the cardinality of a maximal element of $\M_n^1$, and then gave a characterization of the corresponding maximal families.
	
	Clearly, Problem~\ref{pb:state} is still far from being completely solved, and there are several avenues for future research on this subject. We plan to address the following questions in an extended version of this abstract:
	\begin{itemize}
		\item The most natural direction to explore is to generalize our counting results to a larger degree $d$ of the pairwise GCD and to larger finite fields. Clearly, since the higher is $d$ the lower becomes the minimum distance of the resulting subspace code, it makes sense to consider only small values of $d$ for practical applications.
		\item It would be interesting to investigate more closely the trade-off between the cardinality of the maximal families and the degree of the GCD. This will give useful information for the design of subspace codes based on linear CA depending on the application's requirements.
		\item Finally, after determining their size, an interesting research question is to count the \emph{number} of maximal families, by varying the degree of the GCD and the order of the ground field.
	\end{itemize}
	
	\bibliographystyle{abbrv}
	\bibliography{references}

\begin{thebibliography}{1}

\bibitem{gadouleau23}
M.~Gadouleau, L.~Mariot, and S.~Picek.
\newblock Bent functions in the partial spread class generated by linear
  recurring sequences.
\newblock {\em Des. Codes Cryptogr.}, 91(1):63--82, 2023.

\bibitem{gauss-irr}
C.~F. Gau{\ss}.
\newblock {\em Disquisitiones arithmeticae}.
\newblock Humboldt-Universit\"{a}t zu Berlin, 1801.

\bibitem{koetter08}
R.~Koetter and F.~R. Kschischang.
\newblock Coding for errors and erasures in random network coding.
\newblock {\em {IEEE} Trans. Inf. Theory}, 54(8):3579--3591, 2008.

\bibitem{mariot20}
L.~Mariot, M.~Gadouleau, E.~Formenti, and A.~Leporati.
\newblock Mutually orthogonal latin squares based on cellular automata.
\newblock {\em Des. Codes Cryptogr.}, 88(2):391--411, 2020.

\bibitem{mm-automata2023}
L.~Mariot and F.~Mazzone.
\newblock On the minimum distance of subspace codes generated by linear
  cellular automata.
\newblock In L.~Manzoni, L.~Mariot, and D.~R. Chowdhury, editors, {\em Cellular
  Automata and Discrete Complex Systems - 29th {IFIP} {WG} 1.5 International
  Workshop, {AUTOMATA} 2023, Trieste, Italy, August 30 - September 1, 2023,
  Proceedings}, volume 14152 of {\em Lecture Notes in Computer Science}, pages
  105--119. Springer, 2023.

\bibitem{medard2011}
M.~M{\'e}dard and A.~Sprintson.
\newblock {\em Network coding: Fundamentals and applications}.
\newblock Academic Press, 2011.

\end{thebibliography}
	
\end{document}